\documentclass[reqno,tbtags,intlimits,a4paper,oneside,12pt]{amsart}
\usepackage[cp1251]{inputenc}%
\usepackage[english]{babel}
\usepackage{amssymb,upref,calc,url}
\usepackage[mathscr]{eucal}
\usepackage{graphicx}
\usepackage{amsmath,amscd,amsthm,verbatim}
\usepackage[all]{xy}
\usepackage[in]{fullpage}
\newtheorem{thm}{Theorem}[section]
\newtheorem{lem}[thm]{Lemma}
\newtheorem{cor}[thm]{Corollary}
\newtheorem{prb}[thm]{Problem}

\theoremstyle{definition}
\newtheorem{dfn}{Definition}[section]

\DeclareMathOperator{\wt}{wt}

\begin{document}

\title{Lie algebras of heat operators\\ in nonholonomic frame}
\author{V.\,M.~Buchstaber, E.\,Yu.~Bunkova}
\address{Steklov Mathematical Institute of Russian Academy of Sciences, Moscow, Russia}
\email{bunkova@mi.ras.ru}
\thanks{E.\,Yu.~Bunkova is supported in part by Young Russian Mathematics award}

\begin{abstract}
Lie algebras of systems of $2 g$ graded heat conduction operators $Q_{2k}$, where $k = 0,1, \ldots,2 g-1$,
determining sigma functions $\sigma(z, \lambda)$ of genus $g = 1,2$, and $3$ hyperelliptic curves are constructed. As a corollary, it is found that a system of three operators~$Q_0, Q_2$ and $Q_4$ is already sufficient to determine the sigma functions.  The operator $Q_0$ is the Euler operator, and each of the operators $Q_{2k}$, $k>0$, determines a $g$-dimensional Schr\"odinger equation with quadratic potential in $z$ for a nonholonomic frame of vector fields in $\mathbb{C}^{2g}$ with coordinates $\lambda$.

An analogy of the Cole--Hopf transformation is considered.
It associates with each solution $\varphi(z, \lambda)$ of a linear system of heat equations a system of~nonlinear equations for~the~vector function $\nabla \ln \varphi(z, \lambda)$, where $\nabla$ is the gradient of the function in~$z$.

For any solution $\varphi(z, \lambda)$  of the system of heat equations the graded ring $\mathcal{R}_{\varphi}$ is introduced. It is generated by the logarithmic derivatives of the function $\varphi(z, \lambda)$ of order of~at~least $2$. The Lie algebra of derivations of the ring $\mathcal{R}_{\varphi}$ is presented explicitly. The~interrelation of this Lie algebra with the system of~nonlinear equations is shown.
In the case when $\varphi(z, \lambda) = \sigma(z, \lambda)$, this leads to a known result of constructing Lie algebras of derivations of~hyperellitic functions of genus $g = 1,2,3$.
\end{abstract}

\maketitle

\section*{Introduction} 

In \cite{Nonhol}, systems of $2g$ heat equations in nonholonomic frame are introduced for any hyperelliptic curve of genus $g$. These systems are described using graded operators $Q_0, Q_2, \ldots, Q_{4g-2}$, where $Q_0$ is the Euler operator that defines the weights of variables
$z=(z_1, z_3, \ldots, z_{2g-1})$, $\wt z_k = -k$, and $\lambda = (\lambda_4, \lambda_6, \ldots, \lambda_{4g+2})$,  $\wt \lambda_k = k$, and each of~the~operators $Q_{2k}$, $k>0$, determines a $g$-dimensional Schr\"odinger equation with quadratic potential in $z$ for a nonholonomic frame of vector fields in $\mathbb{C}^{2g}$ with coordinates~$\lambda$.
Based on the classical parameterization of the group ${\rm PSp}(2g, \mathbb{C})$,
an algebra of operators on the space of solutions of this system of equations as well as a ``seed'' primitive solution are constructed in \cite{Nonhol}.
As~a result, the so-called primitive $\mathbb{Z}^{2g}$--invariant solution is obtained by averaging the primitive solution over the lattice $\mathbb{Z}^{2g}$ contained in the operator algebra.
It is proven that the sigma function $\sigma(z, \lambda)$ of the hyperelliptic curve can~be~identified with such a solution.

The present work is devoted to the study of properties of the discussed systems of equations and development of their applications.
We consider systems of $2g$ heat equations that determine sigma functions $\sigma(z,\lambda)$ of the elliptic curve for $g=1$ and of hyperelliptic curves for $g = 2$ and $3$, where $z = (z_1, z_3, \ldots, z_{2g-1})$, and $\lambda = (\lambda_4, \lambda_6, \ldots, \lambda_{4g+2})$ are parameters of the universal curve.
It is shown that in an infinite-dimensional Lie~algebra of linear operators on the ring of smooth functions $\varphi(z,\lambda)$, the operators of this system form a Lie subalgebra $\mathcal{L}_Q$ with $2g$ generators over the ring $\mathbb{Q}[\lambda]$,  considered as the set of~operators of multiplication by~polynomials
$p(\lambda) \in \mathbb{Q}[\lambda]$. The Lie~algebra $\mathcal{L}_Q$ over $\mathbb{C}$ as a polynomial algebra over $\mathbb{Q}[\lambda]$ turned out to be isomorphic to the polynomial Lie algebra over~$\mathbb{Q}[\lambda]$ of vector fields tangent to the discriminant of a hyperelliptic curve in~$\mathbb{C}^{2g}$.
As~a~corollary, it was found that the system defined by the three operators $Q_0$, $Q_2$ and $Q_4$ is already sufficient to determine the general solution of the original system of $2g$ equations.

A transformation is introduced that maps the system of heat equations in $\varphi(z,\lambda)$ into a system of nonlinear equations in $\nabla \ln \varphi(z, \lambda)$, where~$\nabla$ is the gradient of the function in~$z$. This transformation is a multidimensional analogy of the Cole--Hopf transformation, which translates the one-dimensional heat equation into the Burgers equation.

Let $\varphi(z, \lambda)$ be some smooth solution of the system of heat equations. We denote by $\mathcal{R}_{\varphi}$ the graded commutative ring that is generated over $\mathbb{Q}[\lambda]$ by the logarithmic derivatives of $\varphi(z, \lambda)$ of order of at least $2$. We obtain an explicit description of the Lie algebra of~derivations of the ring~$\mathcal{R}_{\varphi}$ and point out the close interrelation of this Lie algebra with our system of nonlinear equations.
The need to obtain effective descriptions of such Lie~algebras of derivations is stimulated by actual problems of describing the dependence on the initial data of solutions to important problems of mathematical physics.
In particular, in the case $\varphi(z, \lambda) = \sigma(z, \lambda)$ 
we obtain the well-known solution to the problem of~constructing the Lie algebra of derivations of hyperellitic functions of genus $g = 1,2,3$.

\section{Nonholonomic frame} \label{S1}

Let $g \in \mathbb{N}$. Following \cite[Section 4]{A}, we consider the space $\mathbb{C}^{2g+1}$ with coordinates  $(\xi_1, \ldots, \xi_{2g+1})$ and we introduce the hyperplane $\Gamma$  given by the equation $\sum_{k=1}^{2g+1} \xi_k = 0$. The permutation group $S_{2g+1}$ of coordinates in $\mathbb{C}^{2g+1}$ corresponds to the action of the group $A_{2g}$ on $\Gamma$. We associate the vector $\xi \in \Gamma$ with the polynomial
\begin{equation} \label{pol1}
\prod_k (x - \xi_k) = x^{2g+1} + \lambda_4 x^{2 g - 1}  + \lambda_6 x^{2 g - 2} + \ldots + \lambda_{4 g} x + \lambda_{4 g + 2},
\end{equation}
where  $\lambda = (\lambda_4, \lambda_6, \ldots, \lambda_{4 g}, \lambda_{4 g + 2}) \in \mathbb{C}^{2g}$.
We identify the orbit space $\Gamma/A_{2g}$ with the space $\mathbb{C}^{2g}$ with coordinates~$\lambda$.
We denote the variety of regular orbits in $\mathbb{C}^{2g}$ by $\mathcal{B}$.
Thus, $\mathcal{B} \subset \mathbb{C}^{2g}$ is the subspace of parameters $\lambda$ such that the polynomial \eqref{pol1} has no multiple roots, and $\mathcal{B} = \mathbb{C}^{2g} \backslash \Sigma$, where $\Sigma$ is the discriminant hypersurface.

The gradient of any $A_{2g}$-invariant polynomial determines a vector field in~$\mathbb{C}^{2g}$ that is tangent to the discriminant $\Sigma$ of the genus $g$ hyperelliptic curve. Choosing a~multiplicative basis in the ring of $A_{2g}$-invariant polynomials, we can construct the corresponding $2g$ polynomial vector fields, which are linearly independent at each point in $\mathcal{B}$. These fields do not commute and determine a \emph{nonholonomic frame} in~$\mathcal{B}$.

In \cite[Section 4]{A} an approach to constructing an infinite-dimensional Lie algebra of such fields based on the convolution of invariants operation is described.
In the present work, we consider the fields
\[
 L_{0}, \quad L_{2}, \quad L_{4}, \quad  \ldots, \quad L_{4 g - 2},
\]
that correspond to the multiplicative basis in the ring of $A_{2g}$-invariants, composed of elementary symmetric functions. The structure polynomials of the convolution of invariants operation in this basis were obtained by D.\,B.\,Fuchs, see~\cite[Section~4]{A}. Note that the \emph{nonholonomic frame} in $\mathcal{B}$ corresponding to the multiplicative basis in the ring of $A_{2g}$-invariants composed of Newton polynomials is used in the works of V.\,M.\,Buchstaber and~A.\,V.\,Mikhailov, see
\cite{BMinf}. 

We express explicitly the vector fields $\{L_{2k}\}$ in the coordinates $(\lambda)$. For convenience, we assume that $\lambda_s = 0$ for all $s \notin \{ 4,6, \ldots, 4 g, 4 g + 2\}$.
For $k, m \in \{ 1, 2, \ldots, 2 g\}$, $k\leqslant m$,~set
\[
 T_{2k, 2m} = 2 (k+m) \lambda_{2 (k+m)} + \sum_{s=2}^{k-1} 2 (k + m - 2 s) \lambda_{2s} \lambda_{2 (k+m-s)}
 - {2 k (2 g - m + 1) \over 2 g + 1} \lambda_{2k} \lambda_{2m},
\] 
and for $k > m $ set $T_{2k, 2m} = T_{2m, 2k}$.

\begin{lem}\label{lem1}
For $k = 0, 1, 2, \ldots, {2 g - 1}$ the formula holds
\begin{equation} \label{Lk}
 L_{2k} = \sum_{s = 2}^{2 g + 1} T_{2k + 2, 2 s - 2} {\partial \over \partial \lambda_{2s}}.
\end{equation} 
\end{lem}

The expressions for the matrix $T = (T_{2k, 2m})$ in \eqref{Lk} are taken from \cite[\S 4]{B3}. A detailed proof of the Lemma can be found in \cite[Lemma 3.1]{4A}.

The vector field $L_0$ coincides with the Euler vector field. For all $k$ we have
\begin{equation} \label{L0}
 [L_0, L_{2k}] = 2 k L_{2k}.
\end{equation}

\subsection{Vector fields tangent to the discriminant  of genus $g=1$ hyperelliptic~curve}

In this case we get vector fields
\begin{align} \label{T1}
L_{0} &= 4 \lambda_4 {\partial \over \partial \lambda_4} + 6 \lambda_6 {\partial \over \partial \lambda_6}, &
L_{2} &= 6 \lambda_6 {\partial \over \partial \lambda_4} - {4 \over 3} \lambda_4^2 {\partial \over \partial \lambda_6},
\end{align}
i.e. the matrix $T = (T_{2k, 2m})$ has the form
\begin{equation}
 T =
 \begin{pmatrix} \label{T10}
 4 \lambda_4 & 6 \lambda_6\\
 6 \lambda_6 & - {4 \over 3} \lambda_4^2
 \end{pmatrix}.
\end{equation}
The commutation relation holds
$[L_0, L_2] = 2 L_2$.

\subsection{Vector fields tangent to the discriminant  of genus $g=2$ hyperelliptic~curve}

In this case the matrix $T = (T_{2k, 2m})$ has the form
\begin{equation} \label{T2}
 T =  \begin{pmatrix}
 4 \lambda_4 & 6 \lambda_6 & 8 \lambda_8 & 10 \lambda_{10}\\
 6 \lambda_6 & 8 \lambda_8 & 10 \lambda_{10} & 0 \\
 8 \lambda_8 & 10 \lambda_{10} & 4 \lambda_4 \lambda_8 & 6 \lambda_4 \lambda_{10} \\
 10 \lambda_{10} & 0 & 6 \lambda_4 \lambda_{10} & 4 \lambda_6 \lambda_{10}
 \end{pmatrix} 
 - {1 \over 5} 
 \begin{pmatrix}
 0 & 0 & 0 & 0 \\
 0 & 12 \lambda_4^2 & 8 \lambda_4 \lambda_6 & 4 \lambda_4 \lambda_8 \\
 0 & 8 \lambda_4 \lambda_6 & 12 \lambda_6^2 & 6 \lambda_6 \lambda_8 \\
 0 & 4 \lambda_4 \lambda_8 & 6 \lambda_6 \lambda_8 & 8 \lambda_8^2 \\
 \end{pmatrix}.
\end{equation}

\begin{lem}[{\cite[Lemma 30]{B2}}]\label{lg2} The commutation relations~\eqref{L0} hold, as well as
\begin{equation} \label{M2}
 \begin{pmatrix}
 [L_2, L_4] \\
 [L_2, L_6] \\
 [L_4, L_6] \\
 \end{pmatrix}
=
\mathcal{M}
 \begin{pmatrix}
 L_0 \\
 L_2 \\
 L_4 \\
 L_6 \\
 \end{pmatrix}, \quad
 \mathcal{M} =  {2 \over 5} \begin{pmatrix}
4 \lambda_6 & - 4 \lambda_4 & 0 & 5 \\
2 \lambda_8 & 0 & - 2 \lambda_4 & 0 \\
- 5 \lambda_{10} & 3 \lambda_8 & - 3 \lambda_6 & 5 \lambda_4 
 \end{pmatrix}.
\end{equation}
\end{lem}

\subsection{Vector fields tangent to the discriminant  of genus $g=3$ hyperelliptic~curve}

In this case the matrix $T = (T_{2k, 2m})$ has the form
\begin{multline} \label{T3}
 T = \begin{pmatrix}
     4 \lambda_4 & 6 \lambda_6 & 8 \lambda_8 & 10 \lambda_{10} & 12 \lambda_{12} & 14 \lambda_{14} \\
     6 \lambda_6 & 8 \lambda_8 & 10 \lambda_{10} & 12 \lambda_{12} & 14 \lambda_{14} & 0 \\
     8 \lambda_8 & 10 \lambda_{10} & 12 \lambda_{12} + 4 \lambda_4 \lambda_8 & 14 \lambda_{14} + 6 \lambda_4 \lambda_{10} & 8 \lambda_4 \lambda_{12} & 10 \lambda_4 \lambda_{14} \\
     10 \lambda_{10} & 12 \lambda_{12} & 14 \lambda_{14} + 6 \lambda_4 \lambda_{10} & 4 \lambda_6 \lambda_{10} + 8 \lambda_4 \lambda_{12}
     & 6 \lambda_6 \lambda_{12} + 10 \lambda_4 \lambda_{14} & 8 \lambda_6 \lambda_{14} \\
     12 \lambda_{12} & 14 \lambda_{14} & 8 \lambda_4 \lambda_{12} & 6 \lambda_6 \lambda_{12} + 10 \lambda_4 \lambda_{14} & 4 \lambda_8 \lambda_{12} + 8 \lambda_6 \lambda_{14}
     & 6 \lambda_8 \lambda_{14} \\
     14 \lambda_{14} & 0 & 10 \lambda_4 \lambda_{14} & 8 \lambda_6 \lambda_{14} & 6 \lambda_8 \lambda_{14} & 4 \lambda_{10} \lambda_{14} \\
     \end{pmatrix} - \\
     - {1 \over 7}
     \begin{pmatrix}
     0 & 0 & 0 & 0 & 0 & 0 \\
     0 & 20 \lambda_4^2 & 16 \lambda_4 \lambda_6 & 12 \lambda_4 \lambda_8 &  8 \lambda_4 \lambda_{10} & 4 \lambda_4 \lambda_{12} \\
     0 & 16 \lambda_4 \lambda_6 & 24 \lambda_6^2 & 18 \lambda_6 \lambda_8 & 12 \lambda_6 \lambda_{10} & 6 \lambda_6 \lambda_{12} \\
     0 & 12 \lambda_4 \lambda_8 & 18 \lambda_6 \lambda_8 & 24 \lambda_8^2 & 16 \lambda_8 \lambda_{10} & 8 \lambda_8 \lambda_{12} \\
     0 &  8 \lambda_4 \lambda_{10} & 12 \lambda_6 \lambda_{10} & 16 \lambda_8 \lambda_{10} & 20 \lambda_{10}^2 & 10 \lambda_{10} \lambda_{12} \\
     0 &  4 \lambda_4 \lambda_{12} &  6 \lambda_6 \lambda_{12} &  8 \lambda_8 \lambda_{12} & 10 \lambda_{10} \lambda_{12} & 12 \lambda_{12}^2 \\
     \end{pmatrix}.  
\end{multline}

\begin{lem}[{\cite[Lemma 4.3]{B3}}]\label{lg3}
The commutation relations~\eqref{L0} hold, as well as
\begin{equation} \label{M3}
 \begin{pmatrix}
 [L_2, L_4] \\
 [L_2, L_6] \\
 [L_2, L_8] \\
 [L_2, L_{10}] \\
 [L_4, L_6] \\
 [L_4, L_8] \\
 [L_4, L_{10}] \\
 [L_6, L_8] \\
 [L_6, L_{10}] \\
 [L_8, L_{10}] 
 \end{pmatrix}
=
\mathcal{M}
 \begin{pmatrix}
 L_0 \\
 L_2 \\
 L_4 \\
 L_6 \\
 L_8 \\
 L_{10}
 \end{pmatrix},
 \mathcal{M} =  {2 \over 7} \begin{pmatrix}
8 \lambda_6 & - 8 \lambda_4 & 0 & 7 & 0 & 0 \\
6 \lambda_8 & 0 & - 6 \lambda_4 & 0 & 14 & 0 \\
4 \lambda_{10} & 0 & 0 & - 4 \lambda_4 & 0 & 21 \\
2 \lambda_{12} & 0 & 0 & 0 & - 2 \lambda_4 & 0 \\
- 7 \lambda_{10} & 9 \lambda_8 & - 9 \lambda_6 & 7 \lambda_4 & 0 & 7 \\
- 14 \lambda_{12} & 6 \lambda_{10} & 0 & - 6 \lambda_6 & 14 \lambda_4 & 0 \\
- 21 \lambda_{14} & 3 \lambda_{12} & 0 & 0 & - 3 \lambda_6 & 21 \lambda_4 \\
- 7 \lambda_{14} & - 7 \lambda_{12} & 8 \lambda_{10} & - 8 \lambda_8 & 7 \lambda_6 & 7 \lambda_4 \\
0 & - 14 \lambda_{14} & 4 \lambda_{12} & 0 & - 4 \lambda_8 & 14 \lambda_6 \\
0 & 0 &- 7 \lambda_{14} & 5 \lambda_{12} & - 5 \lambda_{10} & 7 \lambda_8 \\
 \end{pmatrix}.
\end{equation}
\end{lem}

\vfill
\eject

\section{Heat conduction operators} \label{S2}

We consider the space $\mathbb{C}^{3g}$ with coordinates $(z, \lambda)$, where $z = (z_1, z_3, \ldots, z_{2g-1})$ and $\lambda = (\lambda_4, \lambda_6, \ldots, \lambda_{4g+2})$. We introduce the notation $\partial_k = {\partial \over \partial z_k}$ for odd $k$.

Consider the nonholonomic frame defined by the vector fields $\{L_{2k}\}$. The latter have been introduced in \S \ref{S1}.
Following \cite{Nonhol}, we call \emph{heat conduction operators in nonholonomic frame} $\{L_{2k}\}$ the second-order linear differential operators of the form
\[
Q_{2k} = L_{2k} - H_{2k}, 
\]
where
\begin{equation} \label{e2}
H_{2k} = {1 \over 2} \sum \left( \alpha_{a,b}^{(k)}(\lambda)\partial_a \partial_b + 2 \beta_{a,b}^{(k)}(\lambda)z_a\partial_b + \gamma_{a,b}^{(k)}(\lambda)z_a z_b\right) + \delta^{(k)}(\lambda),
\end{equation}
and the summation is over odd $a, b$ from $1$ to $2g-1$.

\begin{dfn}
The system of equations for $\varphi = \varphi(z, \lambda)$
\begin{equation} \label{e3}
Q_{2k} \varphi = 0
\end{equation}
is called the \emph{system of heat equations}.
\end{dfn}

Let $\Phi$ be the ring of smooth functions in $\mathbb{C}^{3g}$ with coordinates $(z, \lambda)$.
Further in this section we will introduce a specific system of operators $\{Q_{2k}\}$ for $g = 1,2,3$, acting on $\Phi$. We follow the works \cite{Nonhol, BL0, BL, 4A}. In the following sections, we will reveal the fundamental properties of these systems.

Let us note that for the system of operators $\{H_{2k}\}$ in the expression \eqref{e2} the relations hold:
\begin{align*}
&\alpha_{a,b}^{(k)}(\lambda) = \left[ \begin{matrix}
                                      1, \quad a+b = 2 k, \quad a, b \in 2 \mathbb{N} + 1,\\
                                      0, \quad \text{else,} \hfill
                                     \end{matrix}
\right.\\
&\beta_{a,b}^{(k)}(\lambda) \text{ is a linear function in } \lambda,\\
&\gamma_{a,b}^{(k)}(\lambda) \text{ is a quadratic function in } \lambda,\\
&\delta^{(k)}(\lambda) = c_k \lambda_{2k} \text{ for constant } c_k.\\
\end{align*}

\subsection{Heat conduction operators for $g=1$}

In this case, we define the heat operators $Q_0$ and $Q_2$ by the expressions $Q_{2k} = L_{2k} - H_{2k}$, where $L_{2k}$ are given by~\eqref{T1} and
\begin{align*}
H_0 &= z_1 \partial_1 - 1, &
H_2 &= {1 \over 2} \partial_1^2 - {1 \over 6} \lambda_4 z_1^2.
\end{align*}

\subsection{Heat conduction operators for $g=2$}

In this case, we define the heat operators $Q_0$, $Q_2$, $Q_4$, and $Q_6$ by the expressions $Q_{2k} = L_{2k} - H_{2k}$, where $L_{2k}$ are~given by \eqref{Lk} for~\eqref{T2}, and
\begin{align*}
H_0 &= z_1 \partial_1 + 3 z_3 \partial_3 - 3,\\
H_2 &= {1 \over 2} \partial_1^2 - {4 \over 5} \lambda_4 z_3 \partial_1 + z_1 \partial_3 - {3 \over 10} \lambda_4 z_1^2 + \left({3 \over 2} \lambda_8 - {2 \over 5} \lambda_4^2\right) z_3^2,\\
H_4 &= \partial_1 \partial_3 - {6 \over 5} \lambda_6 z_3 \partial_1 + \lambda_4 z_3 \partial_3 - {1 \over 5} \lambda_6 z_1^2 + \lambda_8 z_1 z_3 + \left(3 \lambda_{10} - {3 \over 5} \lambda_4 \lambda_6\right) z_3^2 - \lambda_4,\\
H_6 &= {1 \over 2} \partial_3^2 - {3 \over 5} \lambda_8 z_3 \partial_1 - {1 \over 10} \lambda_8 z_1^2 + 2 \lambda_{10} z_1 z_3 - {3 \over 10} \lambda_4 \lambda_8 z_3^2 - {1 \over 2} \lambda_6.
\end{align*}

\vfill
\eject

\subsection{Heat conduction operators for $g=3$}

In this case, we define the heat operators $Q_0$, $Q_2$, $Q_4$, $Q_6$, $Q_8$, and $Q_{10}$ by the expressions $Q_{2k} = L_{2k} - H_{2k}$, where $L_{2k}$ are given by \eqref{Lk} for \eqref{T3}, and
\begin{align*}
H_0 &=
z_1 \partial_1 + 3 z_3 \partial_3 + 5 z_5 \partial_5 - 6,
\\
H_2 &= {1 \over 2} \partial_1^2 - {8 \over 7} \lambda_4 z_3 \partial_1 + \left(z_1 - {4 \over 7} \lambda_4 z_5\right) \partial_3 + 3 z_3 \partial_5 + \\
& - {5 \over 14} \lambda_4 z_1^2 + \left({3 \over 2} \lambda_8 - {4 \over 7} \lambda_4^2\right) z_3^2 +
 \left({5 \over 2} \lambda_{12} - {2 \over 7} \lambda_4 \lambda_8 \right) z_5^2,\\
H_4 &= \partial_1 \partial_3 - {12 \over 7} \lambda_6 z_3 \partial_1 + \left(\lambda_4 z_3 - {6 \over 7} \lambda_6 z_5\right) \partial_3 + \left(z_1 + 3 \lambda_4 z_5 \right) \partial_5 - {2 \over 7} \lambda_6 z_1^2 + \\
& + \lambda_8 z_1 z_3 
 + \left(3 \lambda_{10} - {6 \over 7} \lambda_4 \lambda_6\right) z_3^2 + 3 \lambda_{12} z_3 z_5 + \left(5  \lambda_{14} - {3 \over 7} \lambda_6 \lambda_8 \right) z_5^2
  - 3 \lambda_4,\\
H_6 &= {1 \over 2} \partial_3^2 + \partial_1 \partial_5 - {9 \over 7} \lambda_8 z_3 \partial_1 - {8 \over 7} \lambda_8 z_5 \partial_3 + \left(\lambda_4 z_3 + 2 \lambda_6 z_5\right) \partial_5 - {3 \over 14} \lambda_8 z_1^2 + 2 \lambda_{10} z_1 z_3 + \\ &  + \left({9 \over 2} \lambda_{12} - {9 \over 14} \lambda_4 \lambda_8\right) z_3^2 + \lambda_{12} z_1 z_5 + 6 \lambda_{14} z_3 z_5 + \left({3 \over 2} \lambda_4 \lambda_{12} - {4 \over 7} \lambda_8^2\right) z_5^2 - 2 \lambda_6,\\
H_8 &= \partial_3 \partial_5 - \left({6 \over 7} \lambda_{10} z_3 - \lambda_{12} z_5\right) \partial_1 - {10 \over 7} \lambda_{10} z_5 \partial_3 + \lambda_8 z_5 \partial_5 - {1 \over 7} \lambda_{10} z_1^2 + 3 \lambda_{12} z_1 z_3 +\\+ &\left( 6 \lambda_{14} - {3 \over 7} \lambda_4 \lambda_{10}\right) z_3^2 + 2 \lambda_{14} z_1 z_5 + \lambda_4 \lambda_{12} z_3 z_5 + \left(3 \lambda_{4} \lambda_{14} + \lambda_6 \lambda_{12} -{5 \over 7} \lambda_8 \lambda_{10}\right) z_5^2 - \lambda_8,\\
H_{10} &= {1 \over 2} \partial_5^2 - \left( {3 \over 7} \lambda_{12} z_3 - 2 \lambda_{14} z_5 \right) \partial_1 - {5 \over 7} \lambda_{12} z_5 \partial_3 - \\ & -{1 \over 14} \lambda_{12} z_1^2 + 4 \lambda_{14} z_1 z_3 - {3 \over 14} \lambda_4 \lambda_{12} z_3^2 + 2 \lambda_4 \lambda_{14} z_3 z_5 + \left(2 \lambda_6 \lambda_{14} - {5 \over 14} \lambda_8 \lambda_{12}\right) z_5^2  
 - {1 \over 2} \lambda_{10}.
\end{align*}

\subsection{The commutation relations} \text{ }

Note that to each $p(\lambda) \in \mathbb{Q}[\lambda]$ one can match the operator of multiplication by~$p(\lambda)$, acting on $\Phi$. This operator will also be denoted by $p(\lambda)$.

\begin{lem} \label{lemlQ}
The relations hold
\begin{equation} \label{comQ}
[Q_{2k}, \lambda_{2s}] = T_{2k+2, 2s-2},  
\end{equation}
where $T_{2k+2, 2s-2} \in \mathbb{Q}[\lambda]$ are given in Lemma \ref{lem1}.
\end{lem}

\begin{proof}
In view of \eqref{e2}, we have $[H_{2k}, \lambda_{2s}] = 0$, hence $[Q_{2k}, \lambda_{2s}] = [L_{2k}, \lambda_{2s}] = T_{2k+2, 2s-2}$.
\end{proof}

\vfill
\eject

\begin{lem} \label{lemQ}
For the operators $\{Q_{2k}\}$ the commutation relations hold
\[
 [Q_0, Q_{2k}] = 2 k Q_{2k};
\]
for $g = 2$ the commutation relations hold (cf. \eqref{M2})
\begin{equation} \label{Q2}
 \begin{pmatrix}
 [Q_2, Q_4] \\
 [Q_2, Q_6] \\
 [Q_4, Q_6] \\
 \end{pmatrix}
=
\mathcal{M}
 \begin{pmatrix}
 Q_0 \\
 Q_2 \\
 Q_4 \\
 Q_6 \\
 \end{pmatrix}, \quad 
 \mathcal{M} =  {2 \over 5} \begin{pmatrix}
4 \lambda_6 & - 4 \lambda_4 & 0 & 5 \\
2 \lambda_8 & 0 & - 2 \lambda_4 & 0 \\
- 5 \lambda_{10} & 3 \lambda_8 & - 3 \lambda_6 & 5 \lambda_4 
 \end{pmatrix};
\end{equation}
and for $g = 3$ the commutation relations hold (cf. \eqref{M3})
\begin{equation} \label{Q3}
 \begin{pmatrix}
 [Q_2, Q_4] \\
 [Q_2, Q_6] \\
 [Q_2, Q_8] \\
 [Q_2, Q_{10}] \\
 [Q_4, Q_6] \\
 [Q_4, Q_8] \\
 [Q_4, Q_{10}] \\
 [Q_6, Q_8] \\
 [Q_6, Q_{10}] \\
 [Q_8, Q_{10}] 
 \end{pmatrix}
=
\mathcal{M}
 \begin{pmatrix}
 Q_0 \\
 Q_2 \\
 Q_4 \\
 Q_6 \\
 Q_8 \\
 Q_{10}
 \end{pmatrix},
 \mathcal{M} =  {2 \over 7} \begin{pmatrix}
8 \lambda_6 & - 8 \lambda_4 & 0 & 7 & 0 & 0 \\
6 \lambda_8 & 0 & - 6 \lambda_4 & 0 & 14 & 0 \\
4 \lambda_{10} & 0 & 0 & - 4 \lambda_4 & 0 & 21 \\
2 \lambda_{12} & 0 & 0 & 0 & - 2 \lambda_4 & 0 \\
- 7 \lambda_{10} & 9 \lambda_8 & - 9 \lambda_6 & 7 \lambda_4 & 0 & 7 \\
- 14 \lambda_{12} & 6 \lambda_{10} & 0 & - 6 \lambda_6 & 14 \lambda_4 & 0 \\
- 21 \lambda_{14} & 3 \lambda_{12} & 0 & 0 & - 3 \lambda_6 & 21 \lambda_4 \\
- 7 \lambda_{14} & - 7 \lambda_{12} & 8 \lambda_{10} & - 8 \lambda_8 & 7 \lambda_6 & 7 \lambda_4 \\
0 & - 14 \lambda_{14} & 4 \lambda_{12} & 0 & - 4 \lambda_8 & 14 \lambda_6 \\
0 & 0 &- 7 \lambda_{14} & 5 \lambda_{12} & - 5 \lambda_{10} & 7 \lambda_8 \\
 \end{pmatrix}.
\end{equation}
\end{lem}
\textsc{The proof} is obtained by direct calculation of the corresponding commutators.

\begin{dfn} \label{def22}
The Lie algebra of heat operators in a nonholonomic frame $\{L_{2k}\}$ is a~Lie~algebra over $\mathbb{C}$, which is additively a free $\mathbb{Q}[\lambda]$-module with generators $\{Q_{2k}\}$ and commutation relations $[Q_{2k}, \lambda_{2s}]$, given by Lemma \ref{lemlQ}, and $[Q_{2k}, Q_{2s}]$, given by~Lemma~\ref{lemQ}. The~commutator of basic elements of the form $p(\lambda) Q_{2k}$ for $p(\lambda) \in \mathbb{Q}[\lambda]$ is calculated according to the Leibniz rule. 
\end{dfn}

We denote by $\mathscr{L}_Q$ the Lie algebra from Definition \ref{def22}.

\begin{thm}
The Lie algebra $\mathscr{L}_Q$ as a Lie~algebra over $\mathbb{C}$ is a Lie~subalgebra of the~Lie~algebra of linear operators on the ring of smooth functions $\Phi$.
\end{thm}
  
\textsc{The proof} follows from the construction of the Lie~algebra $\mathscr{L}_Q$.

\begin{thm}
For $g=2$, the function $\varphi = \varphi(z, \lambda)$ gives a solution of the system of heat equations $Q_{2k} \varphi = 0$, $k = 0,1,2,3$, if and only if
\begin{align*}
Q_0 \varphi &= 0, & Q_2 \varphi &= 0, & Q_4 \varphi &= 0.
\end{align*}
\end{thm}
\begin{proof}
From the relations \eqref{Q2} we get the expression for $Q_6$ as a linear combination of~$[Q_2, Q_4]$, $Q_0$, and $Q_2$, with coefficients in $\mathbb{Q}[\lambda]$.
\end{proof}

\begin{thm}
For $g=3$, the function $\varphi = \varphi(z, \lambda)$ gives a solution of the system of heat equations $Q_{2k} \varphi = 0$, $k = 0,1,2,3,4,5$, if and only if
\begin{align*}
Q_0 \varphi &= 0, & Q_2 \varphi &= 0, & Q_4 \varphi &= 0.
\end{align*}
\end{thm}
\begin{proof}
From the relations \eqref{Q3} we get the expressions for $Q_6$, $Q_8$, and~$Q_{10}$ as a linear combination of $Q_0$, $Q_2$, $Q_4$ and their commutators, with coefficients in $\mathbb{Q}[\lambda]$.
\end{proof}

\vfill

\eject

\section{Graded Lie--Rinehart algebras and polynomial Lie algebras} 

We give the basic definitions following the works \cite{BPol, PMil, Rine}. Set $\lambda = (\lambda_4, \lambda_6, \ldots, \lambda_{4g+2})$.

\begin{dfn}[\cite{Rine}]
Let $R$ be a commutative associative ring with unity, let $A$ be a~commutative $R$-algebra, and let $\mathscr{L}$ be a Lie algebra over $R$. The pair $(A, \mathscr{L})$ is called a~\emph{Lie--Rinehart algebra}, if 
\begin{enumerate}
 \item The Lie algebra $\mathscr{L}$ acts on $A$ by left derivations, i.e.
 \[
  X(ab) = X(a) b + a X(b), \quad \forall a,b \in A, \quad \forall X \in \mathscr{L};
 \]
 \item The Lie algebra $\mathscr{L}$ is a left $A$--module;
 \item The pair $(A, \mathscr{L})$ has to satisfy the following compatibility conditions:
\[
 [X, a Y] = X(a) Y + a [X, Y], \quad (a X)(b) = a (X(b)) \quad \text{for} \quad \forall X, Y \in \mathscr{L}, \quad \forall a, b \in A.
\]
\end{enumerate}
\end{dfn}

Consider a graded $R$--algebra $A = R[\lambda_4, \lambda_6, \ldots, \lambda_{4g+2}]$ of polynomials in $\lambda$ over $R$. The~weights $m_k$ of its generators $\lambda_k$ are generally given as
$
 \wt \lambda_k = m_k$,  $m_k \in \mathbb{Z}.
$
We~use the notation where $m_k = k$.

\begin{dfn}
A Lie--Rinehart algebra $(A, \mathscr{L})$ over the ring $R$ is called\\ a~$(2g, N)$--\emph{polynomial Lie algebra}, if
\begin{enumerate}
 \item $A = R[\lambda_4, \lambda_6, \ldots, \lambda_{4g+2}]$;
 \item $\mathscr{L}$ is a free left module of rank $N$ over a graded polynomial algebra $A$;
 \item $\mathscr{L} = \oplus_{i \in \mathbb{Z}} \mathscr{L}_i$ is a $\mathbb{Z}$-graded Lie algebra,
 \[
  [\mathscr{L}_i, \mathscr{L}_j] \subset \mathscr{L}_{i+j}, \quad i,j \in \mathbb{Z},
 \]
and its grading is consistent with the grading of $A$:
\[
p(\lambda) L \in \mathscr{L}_{i+ \wt p(\lambda)},
\quad \wt(L(q(\lambda))) = \wt(q(\lambda)) + i, \quad L \in \mathscr{L}_i,
\]
where $p(\lambda)$ and $q(\lambda)$ are homogeneous polynomials in $A$ with weights $\wt p(\lambda)$ and~$\wt q(\lambda)$, respectively. 
\end{enumerate}
\end{dfn}

In \cite{BPol, ZS} it was proposed to define a $(2g,N)$--polynomial Lie algebra using its special basis (framing) $\{L_k\}$ of $N$ elements in free left $A$--module $\mathscr{L}$ such that each element of~the~basis is a homogeneous element in the $\mathbb{Z}$--graded Lie algebra $\mathscr{L}$. Herewith the framing should determine a representation of~the $A$--module $\mathscr{L}$ in the differential ring of~the graded polynomial algebra $A$.

In the present work we use the notation $\wt L_k = k$.
We obtain the structural relations
\begin{equation} \label{struct}
[L_i, L_j] = \sum_{k} c_{i,j}^k L_k, \qquad [L_k, \lambda_q] = v_k^q,  
\end{equation}
where $c_{i,j}^k$, $v_i^q$ are homogeneous polynomials in the graded algebra $R[\lambda_4, \lambda_6, \ldots, \lambda_{4 g + 2}]$, and~$\wt c_{i,j}^k = i + j - k$, $\wt v_k^q = k + q$.

\begin{lem} \label{T31}
The vector fields $\{L_{2k}\}$ from \S \ref{S1} for $g = 1,2,3$ determine a framing of the~$(2g,2g)$-polynomial Lie algebra where
$A = \mathbb{Q}[\lambda]$.
\end{lem}

\textsc{The proof} follows from \eqref{L0}, \eqref{T10}, \eqref{T2}, \eqref{T3} and the Lemmas \ref{lg2}-\ref{lg3}. \vspace{1mm}

We denote by $\mathscr{L}_L$ the polynomial Lie algebra from Lemma \ref{T31}.

\begin{thm} \label{T32}
The map $\eta\colon L_{2k} \to Q_{2k}$ gives an isomorphism of the Lie algebras $\mathscr{L}_L$ and $\mathscr{L}_Q$ as $\mathbb{Q}[\lambda]$--modules. 
That is, for the set $\{Q_{2k}\}$ the structure polynomials in the expressions \eqref{struct} coincide with the structure polynomials for the set $\{L_{2k}\}$.
\end{thm}

\begin{proof}
To verify the first formula from \eqref{struct} compare the Lemmas \ref{lg2}, \ref{lg3}, and~\ref{lemQ}. The~second  formula follows from \eqref{comQ}.
\end{proof}

\vfill

\eject

\section{Cole--Hopf transformation} \label{S4} 

The classical Cole--Hopf transformation (see \cite{Cole, Hopf}) turns the linear heat equation
\[
 {\partial \varphi \over \partial t} = \nu {\partial^2 \varphi \over \partial x^2}
\]
into the nonlinear Burgers equation
\begin{equation}\label{Bu0}
{\partial u \over \partial t} = \nu {\partial^2 u \over \partial x^2} - u {\partial u \over \partial x} , 
\end{equation}
where $u = - 2 \nu {\partial \over \partial x} \ln \varphi(x,t)$.

We present an analogy of this result for the case of genus $g=1,2,3$ based on heat operators $\{Q_{2k}\}$, introduced in \S \ref{S2}.
This problem is closely related to the problem of differentiation of Abelian functions over parameters presented in \S \ref{S5}. 

For each $g$, consider the system of heat equations \eqref{e3}
\[
 Q_{2k} \varphi = 0, \quad k = 0, 1, \ldots, 2 g - 1,
\]
for a smooth function $\varphi(z,\lambda)$.
Here we introduce the notation $\psi_{k} = - \partial_{k} \ln \varphi$, where $k \in \{ 1, 3, \ldots, 2 g - 1\}$. We consider a multidimensional analogy of the Cole--Hopf transformation, namely a system of equations for a vector function
$(\psi_{1}, \psi_{3}, \ldots, \psi_{2g-1})$,
equivalent to system \eqref{e3} for $\varphi$.

We will write this system in the form
\[
 \mathcal{L}_{2k} \psi_s = w_{2k,s},
\]
where $\mathcal{L}_{2k}$ are differential operators whose coefficients are linear expressions in $z$ and $(\psi_{1}, \psi_{3}, \ldots, \psi_{2g-1})$ over the ring $\mathbb{Q}[\lambda]$, and $w_{2k,s}$ are the corresponding linear functions over the ring $\mathbb{Q}[\lambda]$ in $z$, $(\psi_{1}, \psi_{3}, \ldots, \psi_{2g-1})$, and the second derivatives of the vector $(\psi_{1}, \psi_{3}, \ldots, \psi_{2g-1})$.

\subsection{Differential operators for $g=1$}
\begin{align*}
\mathcal{L}_0 &= L_0 - z_1 \partial_{1}, 
& \mathcal{L}_2 &= L_2 - \psi_1 \partial_{1}.
\end{align*}
The corresponding expressions are:
\begin{align*}
w_{0,1} &= \psi_1, & w_{2,1} &= {1 \over 2} \psi_{111} - {1 \over 3} \lambda_4 z_1.
\end{align*}

\begin{thm}
For $g = 1$ a solution $\varphi$ of the system of heat equations \eqref{e3} gives a solution $\psi_1 = \partial_1 \ln \varphi$ of the system of nonlinear differential equations that we call an~analog of~the~Burgers equation for $g = 1$:
\begin{equation} \label{Bu1}
\mathcal{L}_{2k} \psi_1 = w_{2k,1}, \qquad k = 0,1. 
\end{equation}
\end{thm}
\textsc{The proof} is the result of direct computation.

In it's expanded form, the system of equations \eqref{Bu1} for the function $u = \psi_1$ takes the~form
\begin{align*}
&L_0 u = z_1 \partial_{1} u + u,
&
&L_2 u = {1 \over 2} \partial_{1}^2 u + u \partial_{1} u - {1 \over 3} \lambda_4 z_1.
\end{align*}
The second equation of the system coincides with equation \eqref{Bu0} for $L_2 = - {\partial \over \partial t}$, $\lambda_4 = 0$, and $\nu = - {1 \over 2}$.

\subsection{Differential operators for $g=2$}
\begin{align*}
\mathcal{L}_0 &= L_0 - z_1 \partial_{1} - 3 z_3 \partial_{3}, &
\mathcal{L}_2 &= L_2 + \left(- \psi_1 + {4 \over 5} \lambda_4 z_3\right) \partial_{1} - z_1 \partial_{3}, \\
\mathcal{L}_6 &= L_6 + {3 \over 5} \lambda_8 z_3 \partial_{1} - \psi_3 \partial_{3}, &
\mathcal{L}_4 &= L_4 + \left(- \psi_3 + {6 \over 5} \lambda_6 z_3\right) \partial_{1} - \left(\psi_1 + \lambda_4 z_3\right) \partial_{3}.
\end{align*}
The corresponding expressions are:
\begin{align*}
w_{0,1} &= \psi_1, & w_{0,3} &= 3 \psi_3, \\ w_{2,1} &= {1 \over 2} \psi_{111} + \psi_3 - {3 \over 5} \lambda_4 z_1, & w_{2,3} &= {1 \over 2} \psi_{113} - {4 \over 5} \lambda_4 \psi_1 + \left(3 \lambda_8 - {4 \over 5} \lambda_4^2\right) z_3,\\
w_{4,1} &= \psi_{113} - {2 \over 5} \lambda_6 z_1 + \lambda_8 z_3, & w_{4,3} &= \psi_{133} - {6 \over 5} \lambda_6 \psi_1 + \lambda_4 \psi_3 + \lambda_8 z_1 + \left( 6 \lambda_{10} + {6 \over 5} \lambda_4 \lambda_6 \right) z_3, \\
w_{6,1} &= {1 \over 2} \psi_{133} - {1 \over 5} \lambda_8 z_1 + 2 \lambda_{10} z_3,
& w_{6,3} &= {1 \over 2} \psi_{333} - {3 \over 5} \lambda_8 \psi_1 + 2 \lambda_{10} z_1 - {3 \over 5} \lambda_4 \lambda_8 z_3.
\end{align*}

\begin{thm}
For $g = 2$ a solution $\varphi$ of the system of heat equations \eqref{e3} gives a~solution $(\psi_1, \psi_3) = (\partial_1 \ln \varphi, \partial_3 \ln \varphi)$ of the system of nonlinear differential equations that we call an analog of the Burgers equation for $g = 2$:
\begin{equation} \label{Bu2}
\mathcal{L}_{2k} (\psi_1, \psi_3) = (w_{2k,1}, w_{2k,3}), \qquad k = 0,1,2,3.
\end{equation}
\end{thm}
\textsc{The proof} is the result of direct computation.

In it's expanded form, the system of equations \eqref{Bu2} takes the form
\begin{align*}
&L_0 \psi_1 = z_1 \partial_{1} \psi_1 + 3 z_3 \partial_{3} \psi_1 + \psi_1, \\
&L_0 \psi_3 = z_1 \partial_{1} \psi_3 + 3 z_3 \partial_{3} \psi_3 + 3 \psi_3, \\
&L_2 \psi_1 = {1 \over 2} \partial_{1}^2 \psi_{1} + \left(\psi_1 - {4 \over 5} \lambda_4 z_3\right) \partial_{1} \psi_1 + z_1 \partial_{3} \psi_1 +  \psi_3 - {3 \over 5} \lambda_4 z_1, \\
&L_2 \psi_3 = {1 \over 2} \partial_{1}^2 \psi_{3} + \left(\psi_1 - {4 \over 5} \lambda_4 z_3\right) \partial_{1} \psi_3 + z_1 \partial_{3} \psi_3 - {4 \over 5} \lambda_4 \psi_1 + \left(3 \lambda_8 - {4 \over 5} \lambda_4^2\right) z_3, \\
&L_4 \psi_1 = \partial_{1} \partial_{3} \psi_{1} + \left(\psi_3 - {6 \over 5} \lambda_6 z_3\right) \partial_{1} \psi_1 + \left(\psi_1 + \lambda_4 z_3\right) \partial_{3} \psi_1 - {2 \over 5} \lambda_6 z_1 + \lambda_8 z_3, \\
&L_4 \psi_3 = \partial_{1} \partial_{3} \psi_{3} + \left(\psi_3 - {6 \over 5} \lambda_6 z_3\right) \partial_{1} \psi_3 + \left(\psi_1 + \lambda_4 z_3\right) \partial_{3} \psi_3 - {6 \over 5} \lambda_6 \psi_1 + \lambda_4 \psi_3 + \\
& \qquad  \qquad \qquad \qquad \qquad \qquad \qquad \qquad \qquad \qquad \qquad + \lambda_8 z_1 + \left( 6 \lambda_{10} + {6 \over 5} \lambda_4 \lambda_6 \right) z_3, \\
&L_6 \psi_1 = {1 \over 2} \partial_{3}^2 \psi_{1} - {3 \over 5} \lambda_8 z_3 \partial_{1} \psi_1 + \psi_3 \partial_{3} \psi_1 - {1 \over 5} \lambda_8 z_1 + 2 \lambda_{10} z_3,\\
&L_6 \psi_3 = {1 \over 2} \partial_{3}^2 \psi_{3} - {3 \over 5} \lambda_8 z_3 \partial_{1} \psi_3 + \psi_3 \partial_{3} \psi_3 - {3 \over 5} \lambda_8 \psi_1 + 2 \lambda_{10} z_1 - {3 \over 5} \lambda_4 \lambda_8 z_3.
\end{align*}

\subsection{Differential operators for $g=3$}
\begin{align*}
\mathcal{L}_0 &= L_0 - z_1 \partial_{1} - 3 z_3 \partial_{3} - 5 z_5 \partial_{5}, \\
\mathcal{L}_2 &= L_2 - \left(\psi_1 - {8 \over 7} \lambda_4 z_3 \right) \partial_{1}
- \left( z_1 - {4 \over 7} \lambda_4 z_5 \right) \partial_{3} - 3 z_3 \partial_{5}, \\
\mathcal{L}_4 &= L_4 - \left(\psi_3 - {12 \over 7} \lambda_6 z_3 \right) \partial_{1}
- \left( \psi_1 + \lambda_4 z_3 - {6 \over 7} \lambda_6 z_5 \right) \partial_{3} - (z_1 + 3 \lambda_4 z_5) \partial_{5}, \\
\mathcal{L}_6 &= L_6 - \left(\psi_5 - {9 \over 7} \lambda_8 z_3 \right) \partial_{1}
- \left(\psi_3 - {8 \over 7} \lambda_8 z_5 \right) \partial_{3}
- \left(\psi_1 + \lambda_4 z_3 + 2 \lambda_6 z_5 \right) \partial_{5}, \\
\mathcal{L}_8 &= L_8 + \left({6 \over 7} \lambda_{10} z_3 - \lambda_{12} z_5\right) \partial_{1}
- \left(\psi_5 - {10 \over 7} \lambda_{10} z_5 \right) \partial_{3}
- \left(\psi_3 + \lambda_8 z_5 \right) \partial_{5}, \\
\mathcal{L}_{10} &= L_{10} + \left( {3 \over 7} \lambda_{12} z_3 - 2 \lambda_{14} z_5 \right) \partial_{1}
+ {5 \over 7} \lambda_{12} z_5 \partial_{3} - \psi_5 \partial_{5}. 
\end{align*}
The corresponding expressions are:
\begin{align*}
w_{0,1} &= \psi_1, \qquad w_{0,3} = 3 \psi_3, \qquad w_{0,5} = 5 \psi_5,\\
w_{2,1} &= {1 \over 2} \psi_{111} + \psi_3 - {5 \over 7} \lambda_4 z_1, \\
w_{2,3} &= {1 \over 2} \psi_{113} - {8 \over 7} \lambda_8 \psi_1 + 3 \psi_5 + \left(3 \lambda_8 - {8 \over 7} \lambda_4^2\right) z_3, \\
w_{2,5} &= {1 \over 2} \psi_{115} - {4 \over 7} \lambda_4 \psi_3 + \left( 5 \lambda_{12} - {4 \over 7} \lambda_4 \lambda_8 \right) z_5,\\
w_{4,1} &= \psi_{113} + \psi_5 - {4 \over 7} \lambda_6 z_1 + \lambda_8 z_3,\\
w_{4,3} &= \psi_{133} - {12 \over 7} \lambda_6 \psi_1 + \lambda_4 \psi_3 + \lambda_8 z_1 + \left(6 \lambda_{10} - {12 \over 7} \lambda_4 \lambda_6\right) z_3 + 3 \lambda_{12} z_5,\\
w_{4,5} &= \psi_{135} - {6 \over 7} \lambda_6 \psi_3 + 3 \lambda_4 \psi_5 + 3 \lambda_{12} z_3 + \left(10 \lambda_{14} - {6 \over 7} \lambda_6 \lambda_8\right) z_5, 
\\
w_{6,1} &= {1 \over 2} \psi_{133} + \psi_{115} - {3 \over 7} \lambda_6 z_1 + 2 \lambda_{10} z_3 + \lambda_{12} z_5,\\
w_{6,3} &= {1 \over 2} \psi_{333} + \psi_{135} - {9 \over 7} \lambda_8 \psi_1 + \lambda_4 \psi_5 + 2 \lambda_{10} z_1 + \left(9 \lambda_{12} - {9 \over 7} \lambda_4 \lambda_8\right) z_3 + 6 \lambda_{14} z_5,\\
w_{6,5} &= {1 \over 2} \psi_{335} + \psi_{155} - {8 \over 7} \lambda_8 \psi_3 + 2 \lambda_6 \psi_5 + \lambda_{12} z_1 + 6 \lambda_{14} z_3 + \left(3 \lambda_4 \lambda_{12} - {8 \over 7} \lambda_8^2 \right) z_5,
\\
w_{8,1} &= \psi_{135} - {2 \over 7} \lambda_{10} z_1 + 3 \lambda_{12} z_3 + 2 \lambda_{14} z_5,\\
w_{8,3} &= \psi_{335} - {6 \over 7} \lambda_{10} \psi_1 + 3 \lambda_{12} z_1 + (12 \lambda_{14} - {6 \over 7} \lambda_4 \lambda_{10}) z_3 + \lambda_4 \lambda_{12} z_5,\\
w_{8,5} &= \psi_{355} + \lambda_{12} \psi_1 - {10 \over 7} \lambda_{10} \psi_3 + \lambda_8 \psi_5 + \\
& \qquad + 2 \lambda_{14} z_1 + \lambda_4 \lambda_{12} z_3 + 2 \left(3 \lambda_4 \lambda_{14} + \lambda_6 \lambda_{12} - {5 \over 7} \lambda_8 \lambda_{10}\right) z_5,
\end{align*}
\begin{align*}
w_{10,1} &= {1 \over 2} \psi_{155} - {1 \over 7} \lambda_{12} z_1 + 4 \lambda_{14} z_3,\\
w_{10,3} &= {1 \over 2} \psi_{355} - {3 \over 7} \lambda_{12} \psi_1 + 4 \lambda_{14} z_1 - {3 \over 7} \lambda_4 \lambda_{12}z_3 + 2 \lambda_4 \lambda_{14} z_5,\\
w_{10,5} &= {1 \over 2} \psi_{555} + 2 \lambda_{14} \psi_1 - {5 \over 7} \lambda_{12} \psi_3 + 2 \lambda_4 \lambda_{14} z_3 + \left(4 \lambda_6 \lambda_{14} - {5 \over 7} \lambda_8 \lambda_{12}\right) z_5. 
\end{align*}

\begin{thm}
For $g = 3$ a solution $\varphi$ of the system of heat equations \eqref{e3} gives a solution $(\psi_1, \psi_3,  \psi_5) = (\partial_1 \ln \varphi, \partial_3 \ln \varphi, \partial_5 \ln \varphi)$ of the system of nonlinear differential equations that we call an analog of the Burgers equation for $g = 3$:
\[
\mathcal{L}_{2k} (\psi_1, \psi_3, \psi_5) = (w_{2k,1}, w_{2k,3}, w_{2k,5}),  \qquad k = 0,1,2,3,4,5.
\]
\end{thm}
\textsc{The proof} is the result of direct computation.

\section{The problem of differentiation for the function ring}

Analogs of the vector fields $\{\mathcal{L}_{2k}\}$ for $g=1,2,3$ were introduced in \cite{FS, B2, B3} in connection with the
problem of differentiation of Abelian functions over parameters, see~\S \ref{S5}.
We formulate an analogy of this problem, which is solved by the fields $\{\mathcal{L}_{2k}\}$.

We introduce a ring of functions $\mathcal{R}_\varphi$. The generators of this graded ring over $\mathbb{Q}[\lambda]$ are the functions
$\psi_{k_1 \ldots k_n} = - \partial_{k_1} \cdots \partial_{k_n} \ln \varphi$,
where $n \geqslant 2$, $k_s \in \{ 1, 3, \ldots, 2 g - 1\}$, and~$\wt \psi_{k_1 \ldots k_n} = k_1 + \ldots + k_n$, $\wt \lambda_k = k$.

\begin{prb} \label{ppp}
Construct the ring of differential operators in $(z, \lambda)$ that are derivations of~$\mathcal{R}_\varphi$.
\end{prb}

\begin{thm} \label{thmL1}
For $g=1$ for the operators $\{\mathcal{L}_{2k}\}$ and $\{\partial_k\}$ the commutation relations hold
\begin{align*}
[\mathcal{L}_0, \partial_1] &= \partial_1; &
[\mathcal{L}_0, \mathcal{L}_2] &= 2 \mathcal{L}_2; &
[\partial_1, \mathcal{L}_2] &= - \psi_{11} \partial_1.
\end{align*}
\end{thm}

\textsc{The proof} is obtained by direct calculation of the corresponding commutators.

\begin{thm} \label{thmL2}
For $g=2$ for the operators $\{\mathcal{L}_{2k}\}$ and $\{\partial_k\}$ the commutation relations hold
\begin{align*}
[\mathcal{L}_0, \mathcal{L}_{2k}] &= 2k \mathcal{L}_{2k}, \quad k=1,2,3; & 
[\mathcal{L}_0, \partial_{k}] &= k \partial_{k}, \quad k=1,3; & [\partial_1, \partial_3] &= 0;
\end{align*}
\begin{align*}
[\partial_1, \mathcal{L}_2] &= - \psi_{11} \partial_1 - \partial_3;
 &
[\partial_1, \mathcal{L}_4] &= -\psi_{13} \partial_1 - \psi_{11} \partial_3; \\
[\partial_1, \mathcal{L}_6] &= -\psi_{13} \partial_3; &
[\partial_3, \mathcal{L}_2] &=  - \left(\psi_{13} - {4 \over 5} \lambda_4 \right) \partial_1; \\
[\partial_3, \mathcal{L}_4] &= -\left(\psi_{33} - {6 \over 5} \lambda_6 \right) \partial_1 - \left(\psi_{13} + \lambda_4\right) \partial_3; &
[\partial_3, \mathcal{L}_6] &= {3 \over 5} \lambda_8 \partial_1 - \psi_{33} \partial_3;
\end{align*}
\begin{align*}
 \begin{pmatrix}
 [\mathcal{L}_2, \mathcal{L}_4] \\
 [\mathcal{L}_2, \mathcal{L}_6] \\
 [\mathcal{L}_4, \mathcal{L}_6] \\
 \end{pmatrix}
&=
\mathcal{M}
 \begin{pmatrix}
\mathcal{L}_0 \\
\mathcal{L}_2 \\
\mathcal{L}_4 \\
\mathcal{L}_6 \\
 \end{pmatrix} + {1 \over 2}
  \begin{pmatrix}
\psi_{113} & -\psi_{111}\\
\psi_{133} & -\psi_{113} \\
\psi_{333} & -\psi_{133}
 \end{pmatrix}
 \begin{pmatrix}
\partial_1 \\
\partial_3 \\
 \end{pmatrix}, \quad \text{where $\mathcal{M}$ is given in \eqref{Q2}}.
\end{align*}
\end{thm}

\textsc{The proof} is similar to the proof of Theorem B.6 from \cite{B2}.

\begin{thm} \label{thmL3}
For $g=3$ for the operators $\{\mathcal{L}_{2k}\}$ and $\{\partial_k\}$ the commutation relations hold
\begin{align*}
[\mathcal{L}_0, \mathcal{L}_{2k}] &= 2k \mathcal{L}_{2k}, \quad k = 1,2,3,4,5; \\
[\mathcal{L}_0, \partial_k] &= k \partial_k, \quad k = 1,3,5;
\\ [\partial_k, \partial_s] &= 0, \quad k,s=1,3,5;
\\
  \begin{pmatrix}
 [\partial_1, \mathcal{L}_2] \\
 [\partial_1, \mathcal{L}_4] \\
 [\partial_1, \mathcal{L}_6] \\
 [\partial_1, \mathcal{L}_8] \\
 [\partial_1, \mathcal{L}_{10}] 
 \end{pmatrix}
&= 
 - \begin{pmatrix}
 \psi_{11} & 1 & 0\\
 \psi_{13} & \psi_{11} & 1 \\
 \psi_{15} & \psi_{13} & \psi_{11} \\
 0 & \psi_{15} & \psi_{13} \\
 0 & 0 & \psi_{15}
 \end{pmatrix}
 \begin{pmatrix}
  \partial_1 \\ \partial_3 \\ \partial_5
 \end{pmatrix};
\\
  \begin{pmatrix}
 [\partial_3, \mathcal{L}_2] \\
 [\partial_3, \mathcal{L}_4] \\
 [\partial_3, \mathcal{L}_6] \\
 [\partial_3, \mathcal{L}_8] \\
 [\partial_3, \mathcal{L}_{10}] 
 \end{pmatrix}
&= 
 - \begin{pmatrix}
 \psi_{13} + \lambda_4 & 0 & 3\\
 \psi_{33} & \psi_{13} + \lambda_4 & 0 \\
 \psi_{35} & \psi_{33} & \psi_{13} + \lambda_4 \\
 0 & \psi_{35} & \psi_{33}\\
 0 & 0 & \psi_{35}
 \end{pmatrix}
 \begin{pmatrix}
  \partial_1 \\ \partial_3\\ \partial_5
 \end{pmatrix} +
 {3 \over 7}
 \begin{pmatrix}
 5 \lambda_4 \\
 4 \lambda_6 \\
 3 \lambda_8 \\
 2 \lambda_{10} \\
 \lambda_{12}
 \end{pmatrix}
 \partial_1;
\\
  \begin{pmatrix}
 [\partial_5, \mathcal{L}_2] \\
 [\partial_5, \mathcal{L}_4] \\
 [\partial_5, \mathcal{L}_6] \\
 [\partial_5, \mathcal{L}_8] \\
 [\partial_5, \mathcal{L}_{10}] 
 \end{pmatrix}
&= 
 - \begin{pmatrix}
 \psi_{15} & 0 & 0\\
 \psi_{35} & \psi_{15} & 0 \\
 \psi_{55} & \psi_{35} &  \psi_{15} \\
 \lambda_{12} & \psi_{55} & \psi_{35} \\
 2 \lambda_{14} & \lambda_{12} &  \psi_{55}
 \end{pmatrix}
 \begin{pmatrix}
  \partial_1 \\ \partial_3\\ \partial_5
 \end{pmatrix} +
 {2 \over 7}
 \begin{pmatrix}
 2 \lambda_{4} \\
 3 \lambda_6 \\
 4 \lambda_8 \\
 5 \lambda_{10} \\
 6 \lambda_{12} 
 \end{pmatrix}
 \partial_3
 - \begin{pmatrix}
 0 \\
 3 \lambda_4 \\
 2 \lambda_6 \\
 \lambda_8 \\
 0
 \end{pmatrix}
 \partial_5;
\\
 \begin{pmatrix}
 [\mathcal{L}_2, \mathcal{L}_4] \\
 [\mathcal{L}_2, \mathcal{L}_6] \\
 [\mathcal{L}_2, \mathcal{L}_8] \\
 [\mathcal{L}_2, \mathcal{L}_{10}] \\
 [\mathcal{L}_4, \mathcal{L}_6] \\
 [\mathcal{L}_4, \mathcal{L}_8] \\
 [\mathcal{L}_4, \mathcal{L}_{10}] \\
 [\mathcal{L}_6, \mathcal{L}_8] \\
 [\mathcal{L}_6, \mathcal{L}_{10}] \\
 [\mathcal{L}_8, \mathcal{L}_{10}] 
 \end{pmatrix}
&=
\mathcal{M}
 \begin{pmatrix}
 \mathcal{L}_0 \\
 \mathcal{L}_2 \\
 \mathcal{L}_4 \\
 \mathcal{L}_6 \\
 \mathcal{L}_8 \\
 \mathcal{L}_{10}
 \end{pmatrix}
+ {1 \over 2}
 \begin{pmatrix}
 \psi_{113} & - \psi_{111} & 0 \\
 \psi_{133} + \psi_{115} & - \psi_{113} & - \psi_{111} \\
 2 \psi_{135} & - \psi_{115} & - \psi_{113} \\
 \psi_{155} & 0 & - \psi_{115} \\
 \psi_{333} & - \psi_{133} + 2  \psi_{115} & - 2 \psi_{113} \\
 2 \psi_{335} & 0 & - 2 \psi_{133} \\
 \psi_{355} & \psi_{155} & - 2  \psi_{135} \\
2 \psi_{355} & - 2 \psi_{155} + \psi_{335} & - \psi_{333} \\
 \psi_{555} & \psi_{355} & - \psi_{335} - \psi_{155} \\
0 & \psi_{555} & - \psi_{355} 
 \end{pmatrix}
 \begin{pmatrix}
 \partial_1 \\
 \partial_3 \\
 \partial_5 \\
 \end{pmatrix},
\end{align*}
where $\mathcal{M}$ is given in \eqref{Q3}. 
\end{thm}
\textsc{The proof} is similar to the proof of Corollary 10.2 from \cite{B3}.

\begin{cor} \label{L5}
The operators $\{\mathcal{L}_{2k}\}$ and $\{\partial_k\}$ for $g = 1,2,3$ give a framing of a Lie--Rinehart algebra, where $A = \mathcal{R}_\varphi$.
\end{cor}

We denote by $\mathscr{L}_\mathcal{L}$ the Lie--Rinehart algebra from  Corollary \ref{L5}.

Consider the Lie--Rinehart algebra $(\mathcal{R}_\varphi, \mathscr{L}_{L}\otimes_{\mathbb{Q}[\lambda]} \mathcal{R}_\varphi)$. Directly from commutation relations of Theorems~\ref{thmL1}, \ref{thmL2}, and \ref{thmL3}, we obtain:

\begin{thm}
The map $\mathcal{L}_{2k} \to L_{2k}$, $\partial_k \to 0$ defines a homomorphism of Lie--Rinehart algebras $\mathscr{L}_{L}\otimes_{\mathbb{Q}[\lambda]} \mathcal{R}_\varphi$ and $\mathscr{L}_\mathcal{L}$ that is identity on $\mathbb{Q}[\lambda]$.
\end{thm}

\begin{cor} \label{seec}
The commutators
$
[\partial_k, \mathcal{L}_{2s}]
$
are decomposed into a linear combination of~operators $\{\partial_k\}$ with coefficients in $\mathcal{R}_\varphi$.
\end{cor}

\begin{cor}
The commutators
$[\mathcal{L}_{2k}, \mathcal{L}_{2s}]$ and $[\partial_{k}, \mathcal{L}_{2s}]$
for all $(k,s)$
can be expressed as linear combinations of~operators $\{\mathcal{L}_{2k}\}$ and $\{\partial_k\}$ with coefficients in $\mathcal{R}_\varphi$.
\end{cor}

\begin{thm} \label{Tdif}
If the function $\varphi$ satisfies the system of heat equations in a nonholonomic frame for genus $g$, then the algebra $\mathscr{L}_\mathcal{L}$ is an algebra of derivations of the ring $\mathcal{R}_\varphi$.
\end{thm}

Thus, the constructed operators solve Problem \ref{ppp}.

\begin{proof}
We apply the operators $\partial_k$ to differential operators from \S \ref{S4}. We use Corollary \ref{seec} and the observation that $\partial_k w_{2s,l} \in \mathcal{R}_\varphi$.
\end{proof}

\section{Connection with the problem of~differentiation of genus $g$ hyperelliptic functions} \label{S5}

For a meromorphic function $f$ in $\mathbb{C}^g$ the vector $\omega \in \mathbb{C}^g$ is a period if $f(z+\omega) = f(z)$ for all~$z \in \mathbb{C}^g$.
If the periods of a meromorphic function $f$ form a lattice $\Gamma$ of rank~$2g$ in~$\mathbb{C}^g$, then $f$ is called an \emph{Abelian function}.
We can say that Abelian functions are meromorphic functions on the complex torus $T^g = \mathbb{C}^g/\Gamma$. We denote the coordinates in~$\mathbb{C}^g$ by $z = (z_1, z_3, \ldots, z_{2g-1})$.
We consider hyperelliptic curves of genus $g$ in~the~model
\[
\mathcal{V}_\lambda = \{(x_2, x_{2g+1})\in\mathbb{C}^2 \colon
x_{2g+1}^2 = x_2^{2g+1} + \lambda_4 x_2^{2 g - 1}  + \lambda_6 x_2^{2 g - 2} + \ldots + \lambda_{4 g} x_2 + \lambda_{4 g + 2}\}. 
\]
The curve depends on the parameters $\lambda = (\lambda_4, \lambda_6, \ldots, \lambda_{4 g}, \lambda_{4 g + 2}) \in \mathbb{C}^{2 g}$.

Let $\mathcal{B} \subset \mathbb{C}^{2g}$ be the subspace of parameters such that the curve $\mathcal{V}_{\lambda}$ is nonsingular for~$\lambda \in \mathcal{B}$.
Then $\mathcal{B} = \mathbb{C}^{2g} \backslash \Sigma$, where $\Sigma$ is the discriminant hypersurface.

A \emph{hyperelliptic function of genus} $g$ (see \cite{B2, BEL18}) is a meromorphic function in $\mathbb{C}^g \times \mathcal{B}$,
such that for each $\lambda \in \mathcal{B}$ it's restriction on $\mathbb{C}^g \times \lambda$
is an Abelian function, where $T^g$ is~the Jacobian $\mathcal{J}_\lambda$ of the curve~$\mathcal{V}_\lambda$.
We denote by $\mathcal{F}$ the field of hyperelliptic functions of~genus~$g$. For the properties of this field, see \cite{BEL18}. 

Let $\mathcal{U}$ be the total space of the bundle $\pi: \mathcal{U} \to \mathcal{B}$ with fiber the Jacobian~$\mathcal{J}_\lambda$ of~the~curve~$\mathcal{V}_\lambda$ over $\lambda \in \mathcal{B}$.
Thus, we can say that hyperelliptic functions of genus~$g$ are~meromorphic functions in~$\mathcal{U}$.
By Dubrovin--Novikov Theorem~\cite{DN}, there is a~birational isomorphism between $\mathcal{U}$ and the complex linear space~$\mathbb{C}^{3g}$.

We will need the theory of hyperelliptic Kleinian functions (see \cite{BEL, BEL-97, BEL-12, Baker}, as~well~as~\cite{WW} for elliptic functions).
Take the coordinates
$(z, \lambda) = (z_1, z_3, \ldots, z_{2 g -1},$ $\lambda_4, \lambda_6, \ldots, \lambda_{4 g}, \lambda_{4 g + 2})$
in $\mathbb{C}^g \times \mathcal{B} \subset \mathbb{C}^{3g}$.
Let $\sigma(z, \lambda)$ be the hyperelliptic sigma function (or elliptic function in~case of genus $g=1$). It is defined on the universal cover of $\mathcal{U}$. As before, we set $\partial_k = {\partial \over \partial z_k}$.
Following \cite{B2, BEL18, B3}, we use the notation
\[
\zeta_{k} = \partial_k \ln \sigma(z, \lambda), \qquad
\wp_{k_1, \ldots, k_n} = - \partial_{k_1} \cdots \partial_{k_n} \ln \sigma(z, \lambda),
\]
where $n \geqslant 2$, $k_s \in \{ 1, 3, \ldots, 2 g - 1\}$. 
The functions $\wp_{k_1, \ldots, k_n}$ are examples of hyperelliptic functions.

For the problem of constructing a Lie algebra of derivations of $\mathcal{F}$, a general approach to the solution was developed in \cite{BL0, BL}. In \cite{FS, B2, B3} an explicit solution to this problem was obtained for $g=1,2,3$.
According to \cite{Nonhol}, the function $\sigma(z, \lambda)$ is a solution to the~system of~heat equations~\eqref{e3}. Thus, we obtain a solution to this problem as a~corollary of~Theorem~\ref{Tdif}:

\begin{thm} 
For $g = 1,2,3$ the algebra $\mathscr{L}_\mathcal{L}$ for $\varphi = \sigma$ is an algebra of derivations of~$\mathcal{F}$.
\end{thm}

\begin{proof}
The ring $\mathcal{R}_\sigma$ lies in $\mathcal{F}$.
The functions $\wp_{1,k}$, $\wp_{1,1,k}$ and $\wp_{1,1,1,k}$ are generators of~the~field~$\mathcal{F}$ of hyperelliptic functions of genus~$g$ (see \cite{BEL-97}, \cite[Theorem 2.1]{Bp}).
Thus,~$\mathcal{R}_\sigma$ contains all the generators of $\mathcal{F}$. By Theorem \ref{Tdif}, the algebra $\mathscr{L}_\mathcal{L}$ for $\varphi = \sigma$ is an algebra of derivations of the ring $\mathcal{R}_\sigma$, and, therefore, of the field $\mathcal{F}$.
\end{proof}

\end{document}